\documentclass[sigconf]{acmart}

\usepackage{amssymb}
\usepackage{amsthm}
\usepackage{amsmath}
\usepackage{graphicx}
\usepackage{color}

\newcommand{\1}{\underline{1}}

 \newtheorem{theo}{Theorem}
 \newtheorem{defin}{Definition}
 \newtheorem{lem}{Lemma}
\newtheorem{prop}{Proposition}
\newtheorem{corol}{Corollary}
    \newtheorem{assump}{Assumption}
    
\title{Global attraction of ODE-based mean field models with hyperexponential job sizes}

\author[B. Van Houdt]{Benny Van Houdt\\
Dept. Mathematics and Computer Science\\
University of Antwerp, Belgium}

\date{}

\begin{document}

\begin{abstract}
Mean field modeling is a popular approach to assess the performance of large scale
computer systems. The evolution of many mean field models is characterized by a set
of ordinary differential equations that have a unique fixed point. In order to
prove that this unique fixed point corresponds to the limit of the stationary measures
of the finite systems, the unique fixed point must be a global attractor. While global
attraction was established for various systems in case of exponential job sizes, it is often
unclear whether these proof techniques can be generalized to non-exponential job sizes.

In this paper we show how simple monotonicity arguments can be used to prove global
attraction for a broad class of ordinary differential equations that capture the evolution
of mean field models with hyperexponential job sizes. 
This class includes both existing as well as previously unstudied 
load balancing schemes
and can be used for systems with either finite or infinite buffers. 

The main novelty of the approach exists in using a Coxian representation for the  hyperexponential job sizes 
and a partial order that is stronger than the componentwise partial order used in 
the exponential case.   
\end{abstract}
\maketitle

\section{Introduction}
Mean field models are a popular technique to assess the performance of large scale (computer) systems.
They have been applied in various areas such as load balancing \cite{arpanStSyst,bramsonLB_QUESTA,martin1999,mitzenmacher2,vvedenskaya3,ying_batches}, 
work stealing \cite{gast2010mean,minnebo2}, caching \cite{gast2015transient}, garbage collection \cite{vanhoudt31,vanhoudt32}, CSMA networks \cite{cecchi2015mean}, bin packing \cite{xie1}, file swarming systems \cite{lin_p2p}, coupon collector problems \cite{massoulie1}, etc.  
In many cases the evolution of the mean field model
is described by a simple set of ordinary differential equations (ODEs) and one can show that this set of
ODEs has a unique fixed point (that may even have a closed form). The main idea behind a mean field approximation is
that the stationary distribution of a single component in the network  should (weakly) converge to (the Dirac measure of)
the fixed point of the ODEs as the number of components $N$ tends to infinity. Therefore the 
fixed point approximates the stationary behavior of any component in a large finite system.

Different approaches exist to prove the convergence of the stationary distributions to the fixed point of the
mean field limit as $N$ tends to infinity. The traditional {\it indirect} method
exists in first proving convergence of the stochastic processes over finite time scales, that is, for any fixed $T$,
one shows that the sample paths of the stochastic processes on $[0,T]$ converge towards the solution of the ODEs on
$[0,T]$ (with the appropriate
initial condition). For this step, one can often rely on Kurtz's theorem \cite{ethier1,mitzenmacher2} or 
the convergence of transition semigroups of Markov processes \cite{arpanStSyst,vvedenskaya3}. 
The second step exists in showing that the stochastic systems with finite $N$ each have a stationary measure
and that this sequence of stationary measures has a limit point (which follows from the 
tightness of the stationary measures). The final step then exists in showing that the fixed point is a global attractor and that
the limit point of the stationary measures must be the Dirac measure associated with the fixed point.  
It is fair to state that, given existing mean field theory, proving global attraction of the fixed point is 
often the most demanding step (especially if the state space is a subset of $\mathbb{R}^n$, see Corollary  \ref{cor:Kurtz} in Section \ref{sec:global}). 

A recent {\it direct} method to prove convergence is to rely on Theorem 1 of \cite{ying2016rate} or Theorem 3.2 of \cite{Gast_sigm2017} that
were both obtained using Stein's method \cite{braverman2017stein2}.
This approach does not require proving convergence over finite time scales $[0,T]$. Instead it makes use of
the solution of the so-called Poisson equation. The solution of this equation is expressed as 
an integral that is only properly defined if the fixed point is a global attractor (that is locally exponentially stable).
Thus, Stein's method, when applied to ODE-based mean field models, also requires global attraction of the
fixed point. In fact 
the main challenge in verifying the conditions needed to apply  Theorem 1 of \cite{ying2016rate} or Theorem 3.2 of \cite{Gast_sigm2017} exists in showing that the 
fixed point is a global attractor. 

One approach to prove global attraction of a set of ODEs towards its fixed point relies on defining a Lyapunov function as in \cite{mitzenmacher2}.
However in general coming up with a suitable Lyapunov function, even in case of exponential job sizes, is highly
challenging. A somewhat more flexible approach, that was applied in \cite{arpanStSyst,martin1999,vvedenskaya3} for systems with exponential job sizes,
relies on monotonicity. It is composed of the following three steps. First, one defines the state space $\Omega$ in such a way that the set of ODEs maintains the componentwise partial order $\leq$ over time. In other words, if $h \leq \tilde h$ in the componentwise ordering, then $h(t) \leq \tilde h(t)$ where $h(t)$ and $\tilde h(t)$  are the unique solutions to the set of ODEs with $h(0)=h$ and $\tilde h(0)= \tilde h$. Next, one shows that for any fixed point $\pi$ and $h \in \Omega$
there exists an $h^{(l)},h^{(u)} \in \Omega$ such that $h^{(l)} \leq h, \pi \leq h^{(u)}$ in the componentwise ordering. Finally, global attraction on $\Omega$
follows by proving attraction for any initial point $h \in \Omega$ for which either $h \leq \pi$ or $\pi \leq h$ in the componentwise ordering.

Although it is easy to generalize ODE-based mean field models with exponential job sizes to hyperexponential job sizes
(or even phase-type distributed job sizes), generalizing this monotonicity approach to establish global attraction appears problematic. 
In this paper we nevertheless demonstrate that for a broad class of
ODE-based mean field models with hyperexponential jobs sizes, one can still rely on such a monotonicity argument. 
In order to do so, we introduce two novel key ideas.  First, we set up the ODE-based mean field model using a Coxian representation of the 
hyperexponential distribution. By using this Coxian representation all jobs necessarily start service in phase one, the service
phase can only increase by one at a time and the service completion rate decreases as the phase increases
(see Section \ref{sec:Coxian}). These three features
are essential to find a partial ordering on $\Omega$ that is preserved by the set of ODEs over time.
Second, we rely on a partial ordering that is stronger than the componentwise ordering used in the
exponential case as the set of ODEs does {\it not} preserve the usual componentwise order over time (see Section \ref{sec:order}). 

Hyperexponential distributions are often used to model highly variable workloads \cite{dror_book}. \
Efficient algorithms to fit a hyperexponential distribution to heavy tailed distributions can be found in \cite{feldman98,haverkort_fitting,riska_fitting,starobinski_fitting}.
The class of hyperexponential distributions is also dense in the set of all distributions with a completely monotone probability density function (pdf) \cite[Theorem 3.2]{feldman98},
such as the Pareto and Weibull distribution. 
A pdf $f$ is completely monotone if all its derivatives exist and $(-1)^n f^{(n)}(t) \geq 0$ for all $t>0$ and $n \geq 1$.

Although various mean field models with non-exponential job sizes have been introduced, e.g., \cite{vasantam1}, most of these
papers only focus on the convergence over finite time scales and the uniqueness of the fixed point. One notable exception is
\cite{bramsonLB_QUESTA} which establishes the convergence of the stationary regime for the classic power-of-d load balancing
scheme under FIFO service and any job size distribution with decreasing hazard rate. Their proof is highly technical,
while the approach taken in this paper is much more elementary.  

Instead of focusing on a single mean field model, we identify a set of sufficient conditions such that our result applies
to any mean field model satisfying these conditions. We demonstrate that these conditions are satisfied by various mean
field models, such
as the classic power-of-d load balancing \cite{mitzenmacher2,vvedenskaya3}, the pull/push strategies studied in \cite{minnebo2}
and the power-of-d choices load balancing with batch sampling. Further, we introduce a class of probability distributions
$\mathcal{C}_0$,  show that the set of hyperexponential distributions is a strict subclass of $\mathcal{C}_0$
and establish global attraction under these sufficient conditions for any job size distribution belonging to the class $\mathcal{C}_0$.
In other words, the main result also holds for some job size distributions that are not hyperexponential distributions. We also theoretically characterize the first three moments that can be matched 
with a distribution belonging to $\mathcal{C}_0$.

The paper is structured as follows. In Section \ref{sec:Coxian} we derive a Coxian representation of a hyperexponential
distribution, define the class of distributions $\mathcal{C}_0$, prove that all hyperexponential distributions
belong to $\mathcal{C}_0$ and characterize the first three moments.   In Section \ref{sec:model} we introduce the general form of the set of ODEs
characterizing the mean field model. Examples are presented in Section \ref{sec:examples}. The state space and partial order that
enable us to use monotonicity arguments are outlined in Section \ref{sec:order}.
 The set of sufficient conditions
and the global attraction theorem  are discussed in Section \ref{sec:global}, where
we also show that convergence of the stationary measures then follows
from existing results for systems with finite buffers. These conditions are verified in
Section \ref{sec:examples2} for the examples presented in Section \ref{sec:examples}. The proof of the global attraction theorem  
is detailed in Section \ref{sec:proof}. Conclusions are drawn in Section \ref{sec:concl}.

\section{Coxian representations}\label{sec:Coxian}
A cumulative distribution function (cdf) $F$ is a hyperexponential distribution if there
exists a set of probabilities $\tilde p_1,\ldots,\tilde p_n$ such that $\sum_{k=1}^n \tilde p_k = 1$
and real numbers $\mu_1,\ldots,\mu_n > 0$ such that $F(t) = 1 - \sum_{k=1}^n \tilde p_k e^{-\mu_k t}$.  
Further, a cdf $F$ is a phase-type distribution if there exists
a non-negative vector $\alpha=(\alpha_1,\ldots,\alpha_n)$ with $\sum_{i=1}^n \alpha_i = 1$
and a matrix $S$ with negative diagonal entries, non-negative off-diagonal entries
and non-positive row sums such that $F(t) = 1 - \alpha e^{St} \1$, where $\1$ is a column vector of ones. 
In which case $(\alpha,S)$ is called a phase-type representation of $F$. It is well known
that the representation $(\alpha,S)$ of a phase-type distribution is not unique \cite{ocinneide1}.

The most natural phase-type representation of a hyperexponential distributions is clearly given by
setting $\alpha = (\tilde p_1,\ldots,\tilde p_n)$ and 
\[S = \begin{bmatrix}
-\mu_1 &  & &  \\
 & -\mu_2 & &\\
 &  & \ddots &  & \\
 &  &  & -\mu_n 
\end{bmatrix}.\]
Thus, it is very natural to use this phase-type representation to define an ODE-based mean field
model for systems with hyperexponential job sizes. However, by doing so it appears hard 
(if not impossible) to introduce a partial ordering on the state space that is preserved by the set of ODEs over time.
We therefore propose to use a different phase-type representation, being the Coxian representation
introduced below. 
Note that the choice of the phase-type representation does {\it not} affect the main performance measures 
of the system, such as the queue length or response time distribution. It obviously does
affect measures such as the joint distribution of the queue length and service phase as different
representations of the same distributions do not even require to have the same number of phases $n$. 

A cdf $F$ is a Coxian distribution if and only if
it has a phase-type representation with $\alpha=(1,0,\ldots,0)$ and
a matrix $S$ of the following form
\begin{align}\label{eq:Q}
S=\begin{bmatrix}
-\mu_1 & p_1 \mu_1 & & & \\
 & -\mu_2 & p_2 \mu_2 & &\\
& & \ddots & \ddots &  & \\
& & & -\mu_{n-1} & p_{n-1} \mu_{n-1} \\
& &  &  & -\mu_n 
\end{bmatrix},
\end{align}
with $\mu_i > 0$ and $0 < p_i < 1$. Thus, $F(t) = 1 - (1,0,\ldots,0) e^{St} \1$.
For ease of presentation we define $p_n=0$. 

Cumani \cite{Cumani1982} showed that any distribution that has a phase type representation $(\alpha,S)$ with $S$ triangular is a 
Coxian distribution (of at most the same order $n$).
Further, as mixtures of Erlang distributions with common scale parameter are triangular, the class of Coxian distributions
is dense on the space of distributions on $\mathbb{R}^+$ \cite[p. 163-164]{tijms1994stochastic}. 
We now introduce a subclass of the set of all Coxian distributions.

\begin{defin}
The class $\mathcal{C}_0$ of distributions on $\mathbb{R}^+$ is defined as the class of distributions with
a Coxian representation such that  $\mu_i(1-p_i)$ is decreasing in $i$.
\end{defin}

In this paper we  prove global attraction for a class of ODE-based mean field models 
where the service time distribution belongs to $\mathcal{C}_0$.
Let $\mathcal{C}_{Cox}$ be the set of all Coxian distributions and $\mathcal{C}_{hexp}$ the set of all hyperexponential distributions,
then clearly $\mathcal{C}_0 \subset \mathcal{C}_{Cox}$ and $\mathcal{C}_{hexp} \subset \mathcal{C}_{Cox}$ (the latter due to Cumani).
We now prove that $\mathcal{C}_{hexp}$ is a strict subclass of $\mathcal{C}_0$.
We first derive a simple explicit expression for the parameters of a Coxian representation
of a hyperexponential distribution. To do so we start with a technical lemma.

\begin{lem}\label{lem:mus}
For any $\ell > k \geq 1$ and $\mu_j \not= \mu_k$ for $j > k$, we have
\[\sum_{i=k+1}^\ell \frac{\prod_{v=k}^{i-1} (\mu_v-\mu_\ell)}{\prod_{j=k+1}^{i} (\mu_j-\mu_k)} = -1.\]
\end{lem}
\begin{proof}
The sum can be written as
\[\frac{\mu_k-\mu_\ell}{\mu_{k+1}-\mu_k} \left( 1+\frac{\mu_{k+1}-\mu_\ell}{\mu_{k+2}-\mu_k} \left(1+\ldots  
\frac{\mu_{\ell-2}-\mu_\ell}{\mu_{\ell-1}-\mu_k} \left( 1+
\frac{\mu_{\ell-1}-\mu_\ell}{\mu_{\ell}-\mu_k} \right)  \right) \right).\]
As $(\mu_{i-1}-\mu_\ell)/(\mu_{i}-\mu_k) [ 1+
(\mu_{i}-\mu_\ell)/(\mu_{\ell}-\mu_k)] = (\mu_{i-1}-\mu_\ell)/(\mu_{\ell}-\mu_k)$ this expression collapses to $-1$.
\end{proof}

\begin{prop}\label{prop:hyper}
Let $F(t) = 1 - \sum_{k=1}^n \tilde p_k e^{-\mu_k t}$ be a hyperexponential distribution and assume
without loss of generality that $\mu_1 > \mu_2 > \ldots > \mu_n > 0$. Then, $F(t)$ has a Coxian representation
with parameters $\mu_i$ and 
\begin{align}\label{eq:pi}
p_i = \frac{\sum_{k=i+1}^n \tilde p_k\prod_{j=1}^{i} (1-\frac{\mu_k}{\mu_j})}{\sum_{k=i}^n \tilde p_k  \prod_{j=1}^{i-1} 
(1-\frac{\mu_k}{\mu_j}) }.
\end{align}
\end{prop}
\begin{proof}
We show that both the hyperexponential and Coxian representations are equivalent by showing that both distributions
have the same Laplace Stieltjes transform $\tilde F(s)$ (LST). These transforms are given by
\[\tilde F_{hexp}(s) = \sum_{k=1}^n \tilde p_k \frac{\mu_k}{s+\mu_k}, \]
and
\[\tilde F_{Cox}(s) = \sum_{i=1}^n (1-p_i) \left( \prod_{j=1}^{i-1} p_j \right) 
\prod_{j=1}^i \frac{\mu_j}{s+\mu_j}, \]
as with probability $(1-p_i) \prod_{j=1}^{i-1} p_j$ we visit the first $i$ phases for the Coxian
representation. Using a partial fraction expansion for $1/\prod_{j=1}^i (s+\mu_j)$, we get
\[\tilde F_{Cox}(s) = \sum_{k=1}^n \left( \sum_{i=k}^n  (1-p_i) \left( \prod_{j=1}^{i-1} p_j \right)  
 \prod_{j=1, j\not= k}^i \frac{\mu_j}{\mu_j-\mu_k}  \right) \frac{\mu_k}{s+\mu_k}. \]
Hence, $\tilde F_{hexp}(s)=\tilde F_{Cox}(s)$ if
\begin{align}\label{eq:pktilde}
\tilde p_k = \sum_{i=k}^n  (1-p_i) \left( \prod_{j=1}^{i-1} p_j \right)  
 \prod_{j=1, j\not= k}^i \frac{\mu_j}{\mu_j-\mu_k}. 
\end{align}
This is a linear system in the unknowns $(1-p_i)\prod_{j=1}^{i-1} p_j$ and we now show that its solution
can be expressed as 
\begin{align}\label{eq:prodpi}
(1-p_i) \prod_{j=1}^{i-1} p_j  &=  \sum_{\ell=i}^n  \tilde p_\ell \frac{\mu_\ell}{\mu_i} \prod_{v=1}^{i-1} (1-\frac{\mu_\ell}{\mu_v}). 
\end{align}
For $i=n$ this is immediate from \eqref{eq:pktilde} with $k=n$ (as $p_n = 0$).
We now apply backward induction on $i$. Assume the result holds for $i=k+1,\ldots,n$. From  \eqref{eq:pktilde} we find
\begin{align}\label{eq:step(a)}
\tilde p_k \prod_{j=1}^{k-1} (1-\frac{\mu_k}{\mu_j}) &= (1-p_k)\prod_{j=1}^{k-1} p_j  \nonumber \\
& + \underbrace{\sum_{i=k+1}^n (1-p_i)\left( \prod_{j=1}^{i-1} p_j \right) \prod_{j=k+1}^i \frac{\mu_j}{\mu_j-\mu_k}}_{(a)}.
\end{align}
Applying induction and switching sums yields that (a) equals
\begin{align*}
& \sum_{\ell=k+1}^n \tilde p_\ell \frac{\mu_\ell}{\mu_k} 
\sum_{i=k+1}^\ell \left(\prod_{v=1}^{i-1} (1-\frac{\mu_\ell}{\mu_v}) \right) \left(\prod_{j=k+1}^i \frac{\mu_j}{\mu_j-\mu_k}\right) \frac{\mu_k}{\mu_i}\\
&= \sum_{\ell=k+1}^n \tilde p_\ell \frac{\mu_\ell}{\mu_k} \prod_{v=1}^{k-1} (1-\frac{\mu_v}{\mu_j}) \\
& \ \ \ \ \ \ \cdot \underbrace{\left(\sum_{i=k+1}^\ell \left(\prod_{v=k}^{i-1} \frac{\mu_v-\mu_\ell}{\mu_v} \right) 
\left(\prod_{j=k+1}^i \frac{\mu_j}{\mu_j-\mu_k}\right) \frac{\mu_k}{\mu_i}\right)}_{(b)}.
\end{align*}
The expression in (b) is equal to $-1$ due to Lemma \ref{lem:mus}, which allows us to conclude that \eqref{eq:prodpi} holds due to
\eqref{eq:step(a)}.
Using \eqref{eq:prodpi}, $\prod_{j=1}^{i-1} p_j = \prod_{j=1}^{i} p_j + (1-p_i)\prod_{j=1}^{i-1} p_j$ and backward induction on $i$, 
we may conclude that $\tilde F_{hexp}(s)=\tilde F_{Cox}(s)$  if
\begin{align}\label{eq:prodpi2}
\prod_{j=1}^i  p_j &= \sum_{k=i}^n  \tilde p_k \prod_{j=1}^i (1-\frac{\mu_k}{\mu_j}) = \sum_{k=i+1}^n  \tilde p_k \prod_{j=1}^i (1-\frac{\mu_k}{\mu_j}).
\end{align}
The expression in \eqref{eq:pi} is now immediate, where we note that $p_i \in (0,1)$ as $\tilde p_k > 0$ and 
$0 < (1-\mu_k/\mu_j) < 1$ for $j < k$.
\end{proof}

The above result may be of separate interest. We now use it to establish the following theorem:

\begin{theo}
The class of hyperexponential distributions $\mathcal{C}_{hexp}$ is a subclass of $\mathcal{C}_0$.
\end{theo}
\begin{proof}
Given Proposition \ref{prop:hyper}, it suffices to show that $(1-p_i)\mu_i$ is decreasing in $i$. 
As
\begin{align*}
\mu_i \sum_{k=i+1}^n &\tilde p_k  \prod_{j=1}^{i} (1-\frac{\mu_k}{\mu_j}) = 
 \sum_{k=i+1}^n \tilde p_k  \mu_i \prod_{j=1}^{i-1} (1-\frac{\mu_k}{\mu_j})  \\
&- \sum_{k=i+1}^n \tilde p_k  \mu_k \prod_{j=1}^{i-1} (1-\frac{\mu_k}{\mu_j}) ,
\end{align*}
one readily obtains from \eqref{eq:pi} that
\begin{align}\label{eq:1minpimui}
(1-p_i)\mu_i = \frac{\sum_{k=i}^n \tilde p_k \mu_k \prod_{j=1}^{i-1}(1- \frac{\mu_k}{\mu_j})}{
\sum_{k=i}^n \tilde p_k \prod_{j=1}^{i-1} (1-\frac{\mu_k}{\mu_j})}.
\end{align}
As $(1-\frac{\mu_k}{\mu_i}) = 0$ for $k=i$, we can start both sums in the expression for 
$(1-p_{i+1})\mu_{i+1}$ in $k=i$. 
Further, as $(1-\frac{\mu_k}{\mu_j}) > 0$ for $k > j$, we can rewrite $(1-p_i)\mu_i > (1-p_{i+1})\mu_{i+1}$ as
\begin{align*}
&\left( \sum_{k=i}^n \tilde p_k \mu_k  \xi_{k,i-1} \right)
\left(\sum_{k=i}^n \tilde p_k \xi_{k,i-1} -
\sum_{k=i}^n \tilde p_k  \frac{\mu_k}{\mu_i} \xi_{k,i-1} \right)> \\
&   \left(\sum_{k=i}^n \tilde p_k \xi_{k,i-1} \right)
\left(\sum_{k=i}^n \tilde p_k \mu_k \xi_{k,i-1} -
\sum_{k=i}^n \tilde p_k  \frac{\mu_k^2}{\mu_i} \xi_{k,i-1} \right),
\end{align*}
where we denoted $\prod_{j=1}^{i-1} (1-\frac{\mu_k}{\mu_j})$ as $\xi_{k,i-1}$.
This can be restated as
\begin{align*}
&\left( \sum_{k=i}^n \tilde p_k \mu_k  \xi_{k,i-1} \right)^2 < \left(\sum_{k=i}^n \tilde p_k \xi_{k,i-1} \right) \left( \sum_{k=i}^n \tilde p_k \mu_k^2 \xi_{k,i-1} \right),
\end{align*}
which is equivalent to
\begin{align*}
&\left( \sum_{k=i}^n \mu_k \frac{ \tilde p_k   \xi_{k,i-1}}{
\sum_{k=i}^n \tilde p_k \xi_{k,i-1}} \right)^2 <
\sum_{k=i}^n \mu_k^2 \frac{\tilde p_k   \xi_{k,i-1}}{ 
\sum_{k=i}^n \tilde p_k \xi_{k,i-1}}.
\end{align*}
By defining $X_i$ such that $P[X_i = \mu_k] = \tilde p_k \xi_{k,i-1}/
\sum_{k=i}^n \tilde p_k \xi_{k,i-1}$, the above inequality holds
as $E[X]^2 < E[X^2]$ for any random variable $X$ (that is not deterministic).
\end{proof}
When $n=2$ one can show that all Coxian distributions with $(1-p_1)\mu_1 > (1-p_2) \mu_2 = \mu_2$ are also hyperexponential distributions.
However for $n> 2$ the example below shows that this is not the case, so the set of hyperexponential distributions $\mathcal{C}_{hexp}$ is a {\it strict} 
subclass of the class $\mathcal{C}_0$. Consider the Coxian distribution with parameters 
$\mu_1 = 1, \mu_2 = 2$, $\mu_3=0.1$, $p_1 = 0.1$ and $p_2 = 0.8$. This distribution belongs to the class $\mathcal{C}_0$. However
using \eqref{eq:pktilde}, we see that its LST is given by
\[\tilde F_{Cox}(s) = \frac{83}{90}\frac{1}{s+1}-\frac{3}{190}\frac{2}{s+2}+\frac{16}{171}\frac{0.1}{s+0.1}.\]
This distribution is not a hyperexponential as $\tilde p_2$ is negative.\\



Let $R_i$ be the expected remaining service time of a job in phase $i$. Clearly, $R_n = 1/\mu_n$
and $R_{i-1} = 1/\mu_{i-1} + p_{i-1} R_i$ for $i=2, \ldots,n$. Without loss of generality we assume
that the mean job size equals one, which implies that $R_1 = 1$ (as all jobs start in phase $1$ and stay there for an 
exponential amount of time). 
For later use, we rewrite this as
\begin{align}\label{eq:Ris}
R_i p_{i-1} \mu_{i-1} = \mu_{i-1} R_{i-1} -1.
\end{align} 
\begin{lem}\label{applem:Ri}
If $\mu_i(1-p_i)$ is decreasing in $i$, we have $R_i > R_{i-1}$, for $i=2,\ldots,n$. 
\end{lem}
\begin{proof}
The proof is presented in Appendix \ref{app:lemRi}.
\end{proof}

\paragraph{Remark:} Coxian distributions are sometimes defined using an alternate $(\alpha,S)$ representation given by
$\alpha=(\alpha_1,\ldots,\alpha_n)$ and
\begin{align}
S=\begin{bmatrix}
-\lambda_1 & \lambda_1 & & & \\
 & -\lambda_2 & \lambda_2 & &\\
& & \ddots & \ddots &  & \\
& & & -\lambda_{n-1} & \lambda_{n-1} \\
& &  &  & -\lambda_n 
\end{bmatrix},
\end{align}
with $\alpha_{n-i+1} = (1-p_i) \prod_{j=1}^{i-1} p_j $ and $\lambda_{n-i+1} = \mu_i$.\\

\subsection{Moment matching}\label{sec:moments}

In this section we study the range of the first three moments that can be matched with
a distribution belonging to class $\mathcal{C}_0$. We first establish that any 
distribution in $\mathcal{C}_0$ has a decreasing hazard rate. 

\begin{prop}\label{prop:DHR}
Any distribution belonging to $\mathcal{C}_0$ has a decreasing hazard rate.
\end{prop}
\begin{proof}
Let $\tau_t$ be the
service phase of a job at time $t$ (given that it started service at time $0$) and $Y$ the job size,
then the hazard rate $h(t)$ at time $t$ can be written as
\begin{align}\label{eq:h(t)}
h(t) = \sum_{i=1}^n P[\tau_t = i | Y > t] \mu_i (1-p_i).
\end{align}
We need to show that $h(s) \geq h(t)$ for $0 \leq s < t$. As the hazard rate $h(t)$ defined
in \eqref{eq:h(t)} can be rewritten as
\begin{align*}
h(t) &= \underbrace{P[\tau_t \geq  1 | Y > t]}_{= 1} \mu_1 (1-p_1) \\
&- \sum_{i=2}^n P[\tau_t \geq i | Y > t] 
\underbrace{(\mu_{i-1} (1-p_{i-1})-\mu_i (1-p_i))}_{ > 0}, 
\end{align*}
we find that $h(s) \geq h(t)$ if 
\[P[\tau_s \geq i | Y > s] \leq P[\tau_t \geq i | Y > t],\]
for $s < t$.
This inequality is immediate after noting that:
\begin{align*}
P[\tau_t  \geq i, Y > t] &\geq P[\tau_s  \geq i, Y > t] \\
&= P[\tau_s  \geq i, Y > s] P[\tau_s  \geq i, Y > t | \tau_s  \geq i, Y > s] \\
&=P[\tau_s  \geq i, Y > s] P[ Y > t | \tau_s  \geq i, Y > s] \\
&\geq P[\tau_s  \geq i, Y > s] P[ Y > t | Y > s] \\
& = P[\tau_s  \geq i, Y > s] P[ Y > t]/P[Y>s],
\end{align*}
where the second inequality is due to the fact that the rate at which a service completion
can occur decreases as the phase increases.
\end{proof}
Let $m_i = E[Y^i]$ be the $i$-th moment of the job size distribution $Y$.
For any phase type distribution with representation $(\alpha,S)$ we have
$m_i = i! \alpha (-S)^{-i} \1$.
In order to characterize the set of the first three moments that can be
matched by the distributions belonging to $\mathcal{C}_0$, we focus on
the second and third normalized moments:
\[n_2 = \frac{m_2}{m_1^2}, \ \ \ \ n_3 = \frac{m_3}{m_1 m_2}, \]
where $n_2, n_3 \geq 1$ for any positive valued distribution \cite{Osogami_PEVA}.
The advantage of using the normalized moments is that we no longer need to care about
the first moment. Indeed, if $(\alpha,S)$ matches $n_2$ and $n_3$ and has mean $1$,
then $(\alpha,S/m_1)$ still matches $n_2$ and $n_3$ and has mean $m_1$ (as dividing $S$
by $m_1$ changes the $i$-th moment by a factor $m_1^i$, which implies that $n_2$ and $n_3$ 
are not affected by dividing $S$ by $m_1$). 
Thus, if we found a distribution with mean $1$ in $\mathcal{C}_0$ that matches
$n_2$ and $n_3$, we can simply multiply the rates $\mu_i$ by $1/m_1$ to 
get any desired mean $m_1$. 

Let $\mathcal{A}_{n_2,n_3}^{(n)}$ be the set of normalized second and third moments
that can be matched with a distribution belonging to $\mathcal{C}_0$ with at most $n$ phases.
 
\begin{prop}
The set $\mathcal{A}_{n_2,n_3}^{(2)} = \{(n_2,n_3) | n_2 > 2, n_3 > \frac{3}{2}n_2\} \cup \{(2,3)\}$
\end{prop}
\begin{proof}
By Proposition \ref{prop:DHR} any distribution part of $\mathcal{C}_0$ has a decreasing
hazard rate and therefore its squared coefficient of variation $C_X \geq 1$ \cite[p. 16-19]{stoyan1983comparison}. As $n_2 = 1 + C_X$, we have $n_2 \geq 2$.
Further if $n_2 = 2$, the distribution is the exponential distribution and $n_3$ therefore
equals $3$. For $n_2 > 2$ the value of $n_3$ must exceed $3n_2/2$ as 
Theorem 3.1 in \cite{bobbio4} indicates that this is the case for any order $2$ Coxian distribution 
with $n_2 > 2$ (see also Theorem 1 in \cite{Osogami_PEVA}). 
Thus it remains to show that $\mathcal{C}_0$ contains
a distribution that matches $n_2$ and $n_3$ for any $n_2 > 2$ and $n_3 > 3n_2/2 $.

The proposition in Section 3.1 of \cite{whitt_match_hyper} shows that the set of normalized
moments $n_2$ and $n_3$ that can be matched by a hyperexponential distribution
is exactly the set $\mathcal{A}_{n_2,n_3}^{(2)}$ and in such case the matching
can be achieved with just $2$ phases. In fact, the parameters 
of a two phase hyperexponential distribution that matches $n_2$ and $n_3$ are given by (3.5) and (3.6) in \cite{whitt_match_hyper}.
As all hyperexponential distributions belong to $\mathcal{C}_0$, this completes the proof. 
\end{proof}

We note that the proposition in Section 3.1  in \cite{whitt_match_hyper} indicates that we cannot
match a larger range of $(n_2,n_3)$ values by using more than two phases in
case we restrict ourselves to hyperexponential distributions. 
If we consider Coxian distributions with $n$ phases and
$n_2 > 2$, then Theorem 3.1 in \cite{bobbio4} indicates that we can match 
any $n_3 > (n+1)n_2/n$. Thus, the larger $n$, the lower $n_3$ can become, 
contrary to the class of hyperexponential distributions.

The next proposition shows that while class $\mathcal{C}_0$ lies somewhere between the
class of hyperexponential and Coxian distributions with $n_2 \geq 2$, increasing the
number of phases does {\it not} allow us to match a larger range of $n_3$ values. Thus,
as far as matching the first three moments is concerned, the class $\mathcal{C}_0$
does not provide more flexibility than the set of hyperexponential distributions.

 \begin{prop}
The set $\mathcal{A}_{n_2,n_3}^{(n)} = \mathcal{A}_{n_2,n_3}^{(2)}$ for any $n \geq 2$.
\end{prop}
\begin{proof}
We use induction on $n$ and note that the result clearly holds for $n=2$. The proof follows the same
line of reasoning as the proof of Theorem 3.1 in \cite{bobbio4}. 
Let  $\hat n_2$ and $\hat n_3$ be the normalized moments of a distribution 
in  $\mathcal{C}_0$ represented by $(\alpha,S)$. Denote the matrix $S$ as
\[ S = \begin{bmatrix}
-\mu_1 & p_1 \mu_1 \alpha_{n-1} \\ 0 & A
\end{bmatrix},\]
where $\alpha_{n-1}$ is the first row of the size $n-1$ identity matrix.
Note that $(\alpha_{n-1},A)$ is  a phase type representation of a distribution
with $n-1$ phases in $\mathcal{C}_0$. Let $n_2$ and $n_3$ be the normalized moments of $(\alpha_{n-1},A)$
and $m_1$ be its mean.
By induction we know $n_3 > 3n_2/2$ for $n_2 > 2$.
Using exactly the same arguments as in the proof of Theorem 3.1 in \cite{bobbio4}, we find
\[ \hat n_3 = \frac{3}{g}+\frac{(\hat n_2 g-2)^2}{(g-1)g\hat n_2} \frac{n_3}{n_2},\]
 with $g=1+\mu_1 m_1 p_1 \geq 1$. By induction we know $n_3/n_2 \geq 3/2$, meaning
\begin{align}\label{eq:n3bound}
 \hat n_3 \geq \frac{3}{g}+\frac{(\hat n_2 g-2)^2}{(g-1)g\hat n_2} \frac{3}{2}.
\end{align}
Further in  the proof of Theorem 3.1 in \cite{bobbio4} it is shown that
 $(\hat n_2 g-2)^2/((g-1)g\hat n_2)$ is decreasing in $g$ on $(1,\infty)$ 
whenever $\hat n_2 \geq (n+4)/(n+1)$.
Since $\hat n_2 \geq 2$ as any distribution in $\mathcal{C}_0$ has a 
decreasing hazard rate, we obtain a lower
bound for $\hat n_3$ by taking the limit of $g$ to infinity in \eqref{eq:n3bound}. 
This limit clearly equals $3 \hat n_2/2$, which completes the proof. 
\end{proof}

\section{The form of the ODE}\label{sec:model}
ODE-based mean field models of systems with exponential job sizes (with mean $1$) are often of the 
following form
(see Section \ref{sec:examples} for examples):
\begin{align}\label{eq:ODEexpo}
\frac{d}{dt}h_{\ell,1}(t) &=  f_{\ell,1}(h(t)) - (h_{\ell,1}(t)-h_{\ell+1,1}(t)), 
\end{align}
where $h_{\ell,1}(t)$ represents the fraction of the servers with at least $\ell$ jobs
and $f_{\ell,1}(h(t))$ captures events such as job arrivals and job transfers (see Section \ref{sec:Pullex}).
The term $-(h_{\ell,1}(t)-h_{\ell+1,1}(t))$ reflects the drift due to the exponential
service completions.
 The assumption that the mean job size equals $1$ is made throughout the paper
(without loss of generality). 

We now generalize this set of ODEs to the case where the job sizes belong to class $\mathcal{C}_0$
given that the service discipline is first-come-first-served (FCFS).
Define $h_{\ell,i}(t)$, for $\ell > 0$ and $i=1,\ldots,n$, as the fraction of the queues at time $t$
with a queue length of at least $\ell$ in service phase $j \geq i$.
Thus, $(h_{\ell,i}(t)-h_{\ell,i+1}(t))$ is the fraction of queues at time $t$
with $\ell$ or more jobs that are in service phase $i$. 
For ease of notation let $h_{\ell,n+1}(t)=0$ and $h_{0,1}(t) = 1$.
Note that a service completion in a queue with a length of at least $\ell$ always
decreases $h_{\ell,i}(t)$ for $i\geq 2$ as the next customer starts service in phase $1$, whereas $h_{\ell,1}(t)$
only decreases if the queue length is exactly $\ell$.
Hence, the set of ODEs given by \eqref{eq:ODEexpo} then generalizes to:
\begin{align}\label{eq:newODE1}
\frac{d}{dt}&h_{\ell,1}(t) =  f_{\ell,1}(h(t)) \nonumber \\
& -\sum_{j=1}^n \left[(h_{\ell,j}(t)-h_{\ell,j+1}(t))-(h_{\ell+1,j}(t)-h_{\ell+1,j+1}(t))\right] 
\mu_j (1-p_j) \\
\frac{d}{dt}&h_{\ell,i}(t) = 1[\ell > 1] f_{\ell,i}(h(t)) + (h_{\ell,i-1}(t)-h_{\ell,i}(t))p_{i-1}\mu_{i-1}
\nonumber \\ &-  \sum_{j=i}^n (h_{\ell,j}(t)-h_{\ell,j+1}(t)) \mu_j (1-p_j) \label{eq:newODE2}
\end{align}
for $\ell \geq 1$ and $i=2,\ldots,n$, where the sums are due to service completions and the second term in the drift of $h_{\ell,i}(t)$
corresponds to phase changes. 

We remark that we can also model systems with a finite buffer of size $B$ by
setting $f_{\ell,i}(h) = 0$, for $i=1,\ldots,n$ and $\ell > B$, as this implies
that $h_{\ell,i}(t) = 0$ for $i=1,\ldots,n$ and $\ell > B$.

\section{Examples}\label{sec:examples}
\subsection{JSQ(d): Join-the-Shortest-Queue among d randomly selected servers}

Let us first consider the classic power-of-d choices load balancer \cite{mitzenmacher2,vvedenskaya3}, where jobs
arrive at rate $\lambda N$ to a dispatcher who  immediately assigns incoming jobs among the $N$ servers by
routing the job to the server with the least number of jobs among $d$ randomly selected servers.
In this case 
the function $f$ reflects the changes due to arrivals and one finds for $\ell \geq 1$
\begin{align*} 
f_{\ell,1}(h(t)) &= \lambda (h_{\ell-1,1}(t)^d-h_{\ell,1}(t)^d), 
\end{align*}
as $h_{\ell-1,1}(t)^d-h_{\ell,1}(t)^d$ is the probability that the server with the least number of jobs among
$d$ randomly selected servers has queue length $\ell-1$. Further, 
since the dispatcher does not take the service phase into account when dispatching jobs and
$(h_{\ell-1,i}(t)-h_{\ell,i}(t))/(h_{\ell-1,1}(t)-h_{\ell,1}(t))$ is the probability that a server of length $\ell-1$
is in service phase $j \geq i$, we have
\begin{align*}
f_{\ell,i}(h(t)) &=  f_{\ell,1}(h(t)) \frac{h_{\ell-1,i}(t)-h_{\ell,i}(t)}{h_{\ell-1,1}(t)-h_{\ell,1}(t)}\\
&=  \lambda \left( \sum_{j=0}^{d-1} h_{\ell-1,1}(t)^j h_{\ell,1}(t)^{d-1-j} \right) (h_{\ell-1,i}(t)-h_{\ell,i}(t)), 
\end{align*}
for $\ell > 1$ and $i=2,\ldots,n$, as $(a^d-b^d)/(a-b) = \sum_{j=0}^{d-1} a^j b^{d-1-j}$.
For convenience we also define $f_{1,i}(h(t))=0$ for $i=2,\ldots,n$.


\subsection{Pull and push strategies}\label{sec:Pullex}

In this example we consider the system analyzed in \cite{minnebo2}. It consists of $N$ servers that
each have local job arrivals with rate $\lambda$. Servers that are idle generate probe messages at
rate $r$. A probe message is sent to a random server and if this server has pending jobs, a job
is transferred to the idle server.  
The function $f$ now captures the changes due to arrivals as well as job transfers, hence
\begin{align*} 
f_{\ell,1}&(h(t)) = \lambda (h_{\ell-1,1}(t)-h_{\ell,1}(t)) \nonumber \\
& + r(1-h_{1,1}(t)) [1[\ell = 1] h_{2,1}(t)   - 1[\ell > 1](h_{\ell,1}(t)-h_{\ell+1,1}(t)) ],
\end{align*}
for $\ell \geq 1$ and
\begin{align*}
f_{\ell,i}(h(t)) =   \lambda (h_{\ell-1,i}(t)&-h_{\ell,i}(t)) \\&- 
r(1-h_{1,1}(t)) (h_{\ell,i}(t)-h_{\ell+1,i}(t)),
\end{align*}
for $\ell > 1$ and $i=2,\ldots,n$.
Note that $r(1-h_{1,1}(t)) (h_{\ell,i}(t)-h_{\ell+1,i}(t))$ is the rate at which jobs are
transferred from a server with length $\ell$ in phase $j \geq i$ to an idle server.
Therefore $r(1-h_{1,1}(t)) h_{2,1}(t)$ is the rate at which idle servers become busy due to the
probe messages.

\subsection{JSQ(K,d): Join-the-Shortest-K-Queues among d randomly selected servers} 
This example is a generalization of the first example. Jobs now arrive in batches of size $K$ 
and the dispatcher assigns the $K$ jobs (with independent sizes) belonging to the same batch to the $K$ servers with 
the least number of jobs among $d$ randomly selected servers (with $K \leq d$). This load balancing scheme is called {\it batch sampling} 
in \cite{ying_batches}. The mean field model in \cite{ying_batches} is however different than the one presented
here, as we assume that both $K$ and $d$ are fixed, i.e., do not grow as a function of $N$. 

In this case $\lambda < 1/K$ in order to have
a stable system (as the mean service time of a job equals $1$) and 
the function $f$ once more reflects the changes due to arrivals. Note that 
\begin{align*} 
p_{k,\ell}(h(t)) = \sum_{s=0}^{k-1} {d \choose s} (1-h_{\ell,1}(t))^s h_{\ell,1}(t)^{d-s},
\end{align*}
is the probability that the $k$-th shortest queue has a length of at least $\ell$. As such
\begin{align*} 
f_{\ell,1}& (h(t)) = \lambda \sum_{k=1}^K (p_{k,\ell-1}(h(t))-p_{k,\ell}(h(t))) = \lambda \sum_{s=0}^{K-1} (K-s) {d \choose s} \\
& \cdot \left( (1-h_{\ell-1,1}(t))^s h_{\ell-1,1}(t)^{d-s}- (1-h_{\ell,1}(t))^s h_{\ell,1}(t)^{d-s} \right),
\end{align*}
for $\ell \geq 1$ and
\begin{align*} 
f_{\ell,i}(h(t)) &=  f_{\ell,1}(h(t)) \frac{h_{\ell-1,i}(t)-h_{\ell,i}(t)}{h_{\ell-1,1}(t)-h_{\ell,1}(t)}.
\end{align*}
for $\ell > 1$. In addition we define $f_{1,i}(h(t))=0$ for $i=2,\ldots,n$.

Note that in the special case where $K = d$, one finds that $f_{\ell,1}(h(t))$ simplifies to 
$\lambda K (h_{\ell-1,1}(t)-h_{\ell,1}(t))$. Thus, when $K=d$ the set of ODEs describes the
transient evolution of an $M/Cox/1$ queue with arrival rate $\lambda K$.   

\section{State space and partial order}\label{sec:order}

In the case of exponential job sizes the state space is typically defined as 
\[\Omega_{expo} = \{(h_{\ell,1})_{\ell>0} | 0 \leq h_{\ell,1} \leq 1, h_{\ell+1,1} \leq h_{\ell,1}, \sum_{\ell} h_{\ell,1} < \infty \}, \]
where $h_{\ell,1}$ represents the fraction of queues with length $\ell$ or more. 
The partial order used to prove
global attraction on $\Omega_{expo}$ in case of exponential job sizes is the componentwise order. In this
section we introduce the state space and partial order needed in case of a job size distribution belonging to class
$\mathcal{C}_0$.

We define the state space $\Omega$ of the mean field model in terms of the variables $h_{\ell,i}$ as follows
\begin{align*}
\Omega =\{(h_{\ell,i})_{\ell>0,i\in \{1,\ldots,n\}}| 0 \leq h_{\ell,i} \leq 1, h_{\ell,i+1} \leq h_{\ell,i}, h_{\ell+1,i} \leq h_{\ell,i}, \\
h_{\ell,i}+h_{\ell+1,i+1} \geq h_{\ell+1,i}+h_{\ell,i+1}, \sum_{\ell} h_{\ell,1} < \infty \}.
\end{align*} 
The conditions $ h_{\ell,i+1} \leq h_{\ell,i}$ and $h_{\ell+1,i} \leq h_{\ell,i}$ are obvious as $h_{\ell,i}$ is the fraction
of servers with at least $\ell$ jobs in service phase $j \geq i$. The inequality $h_{\ell,i}+h_{\ell+1,i+1} \geq h_{\ell+1,i}+h_{\ell,i+1}$
may seem a bit unexpected. This inequality can be understood by noting that $\Omega$ corresponds to 
\begin{align*}\bar \Omega =\{(x_0,(x_{\ell,i})_{\ell>0,i\in \{1,\ldots,n\}})&|
 x_0 \geq 0, x_{\ell,i} \geq 0, \\ 
& x_0+\sum_{\ell,i} x_{\ell,i}=1, \sum_{\ell,i} \ell x_{\ell,i} < \infty \}.
\end{align*} 
after a change of variables (i.e., $h_{\ell,i} = \sum_{\ell'\geq \ell} \sum_{i'\geq i} x_{\ell',i'}$),
where $x_{\ell,i}$ is the fraction of servers with exactly $\ell$ jobs in service phase $i$.
Therefore the inequality $h_{\ell,i}+h_{\ell+1,i+1} \geq h_{\ell+1,i}+h_{\ell,i+1}$ follows  from
the fact that $(h_{\ell,i}-h_{\ell,i+1})-(h_{\ell+1,i}-h_{\ell+1,i+1})=x_{\ell,i} \geq 0$.

In the case of a system with a finite buffer of size $B$ the state space reduces to
\begin{align*}
\Omega_B =\{(h_{\ell,i})_{\ell \in \{1,\ldots,B\},i\in \{1,\ldots,n\}}| 0 \leq h_{\ell,i} \leq 1,  h_{\ell,i+1} \leq h_{\ell,i}, \\
 h_{\ell+1,i} \leq h_{\ell,i}, h_{\ell,i}+h_{\ell+1,i+1} \geq h_{\ell+1,i}+h_{\ell,i+1} \}.
\end{align*} 
Whenever the buffer size is finite, we can replace $\Omega$ in all subsequent statements by $\Omega_B$.

\begin{prop}\label{lem:anypi}
For any fixed point $\pi \in \Omega$ of the set of ODEs  given by (\ref{eq:newODE1}-\ref{eq:newODE2}), we
have $\pi_{1,i} = \pi_{1,1} \sum_{j=i}^n \frac{1}{\mu_j}\prod_{s=1}^{j-1} p_s$, for $i=1,\ldots,n$,
where $\mu_i$ and $p_s$ are the parameters of the Coxian representation.
\end{prop}
\begin{proof}
See Appendix \ref{app:lemanypi}.
\end{proof}

We introduce the following partial order on $\Omega$ which reduces to the usual componentwise order
in case of exponential job sizes (i.e., when $n=1$). 
\begin{defin}[partial order $\leq_C$]
Let $h, \tilde h \in \Omega$. We state that $h \leq_{C} \tilde h$ if and only if 
\begin{align}\label{eq:order1}
 h_{\ell,i} &\leq \tilde h_{\ell,i}, 
\end{align}
for all $\ell,i$, and
\begin{align}\label{eq:order2}
h_{\ell_1,1}+\sum_{i=2}^n (h_{\ell_i,i}-h_{\ell_{i-1},i}) 
&\leq \tilde h_{\ell_1,1}+\sum_{i=2}^n (\tilde h_{\ell_i,i}-\tilde h_{\ell_{i-1},i}),
\end{align}
for any set of integers $\ell_1 \geq \ell_2 \geq \ldots \geq \ell_n \geq 1$ with $\ell_1 > \ell_n$.
\end{defin}
It is useful to note that $h_{\ell_1,1}+\sum_{i=2}^n (h_{\ell_i,i}-h_{\ell_{i-1},i})$
is the fraction of the servers for which the queue length is at least $\ell_i$ and the
service phase equals $i$ for some $i \in \{1,\ldots,n\}$.

Without condition \eqref{eq:order2} the order would correspond to the usual componentwise partial order.
To illustrate the need for a stronger partial order, let $n=2$ and consider $h, \tilde h \in \Omega$ with
$h_{1,1}=\tilde h_{1,1}=1$, $h_{1,2}=\tilde h_{1,2}=1/2$, $h_{2,1}=\tilde h_{2,1}=1/2$,
$h_{2,2}=0$, $\tilde h_{2,2}=1/2$ and $h_{3,1}= \tilde h_{3,1}=0$. Thus, in both states
half of the servers have queue length one and the other half has queue length 2.
In state $h$ the servers with length $1$ are in service phase $2$ and the servers with length
$2$ are in phase $1$, while in state $\tilde h$ the phases are reversed (queues with length $i$
are in phase $i$, for $i=1,2$). Note that $h$ is smaller than $\tilde h$ in the componentwise
order, but condition \eqref{eq:order2} is violated with $\ell_1 = 2$ and $\ell_2=1$, meaning
$h \not\leq_C \tilde h$. If we now look at the drift of the number of busy servers due to
service completions, we see that it equals $-\mu_2/2$ in state $h$ and $-\mu_1(1-p_1)/2$ in
state $\tilde h$. Hence, $h_{1,1}=\tilde h_{1,1} = 1$, but $h_{1,1}$ decreases more slowly
than $\tilde h_{1,1}$ (when $\mu_2 < \mu_1 (1-p_1)$). This example therefore shows that the
componentwise partial order used for the set of ODEs with exponential job sizes, is {\it not}
preserved over time by the set of ODEs with a job size distribution in $\mathcal{C}_0$ and we need to replace it by a stronger
partial order, which turns out to be the order $\leq_C$ defined above. 

We end by noting that due to the condition $h_{\ell,i}-h_{\ell+1,i} \geq h_{\ell,i+1}-h_{\ell+1,i+1}$ in $\Omega$, we
have $h_{\ell,1}-h_{\ell+1,1} \geq h_{\ell,i}-h_{\ell+1,i}$ for any $i=2,\ldots,n$ and therefore
\begin{align}\label{eq:helpbound}
 h_{\ell_1,1}+&\sum_{i=2}^n (h_{\ell_i,i}-h_{\ell_{i-1},i}) =  h_{\ell_1,1}+\sum_{i=2}^n \sum_{\ell = \ell_i}^{\ell_{i-1}-1}
(h_{\ell,i}-h_{\ell+1,i}) \nonumber \\
&\leq h_{\ell_1,1}+\sum_{i=2}^n \sum_{\ell = \ell_i}^{\ell_{i-1}-1}
(h_{\ell,1}-h_{\ell+1,1}) \nonumber \\ &= h_{\ell_1,1}+\sum_{i=2}^n (h_{\ell_i,1}-h_{\ell_{i-1},1})= h_{\ell_n,1},
\end{align}
for any $h \in \Omega$.

\section{Global attraction}\label{sec:global}

We now list the assumptions needed to establish the main result. Note that some of the intermediate results do not require
all of the assumptions.

\begin{assump}\label{ass:unique}
The functions $f_{\ell,i}(h):\Omega \rightarrow \mathbb{R}$ are such that for any $h \in \Omega$,  
the set of ODEs given by (\ref{eq:newODE1}-\ref{eq:newODE2}) has a unique solution $h(t):[0,\infty) \rightarrow
\mathbb{R}$ with $h(0)=h$ and
there exists a fixed point $\pi$ in $\Omega$.
\end{assump}

The existence of a unique (continuously differentiable) solution $h(t)$ is guaranteed by defining a norm
on $\mathbb{R}^{\mathbb{N}}$ such that the drift is locally Lipschitz and bounded on $\Omega$.
When the buffer size $B < \infty$, the existence of a fixed point follows
almost immediately from Brouwer's fixed point theorem as $\Omega_B$ is a convex and compact subset of $\mathbb{R}^{Bn}$
and $\Omega_B$ is clearly a forward invariant set \cite{bhatia2002stability}.

The next two assumptions are needed to establish that the partial order $\leq_C$ is preserved over time by
the set of ODEs.

\begin{assump}\label{ass:f}
The functions $f_{\ell,i}(h):\Omega \rightarrow \mathbb{R}$ are non-decreasing in $h_{\ell',i'}$ for
any $(\ell',i')\not=(\ell,i)$. 
\end{assump}

For any set of integers $\ell_1 \geq \ell_2 \geq \ldots \geq \ell_n \geq 1$ with $\ell_1 > \ell_n$,
define $g_{(\ell_1,\ldots,\ell_n)}$ and $f_{(\ell_1,\ldots,\ell_n)}$ as a function from $\Omega$ to $\mathbb{R}$ such that 
\begin{align}
g_{(\ell_1,\ldots,\ell_n)}(h) &= h_{\ell_1,1}+\sum_{i=2}^n (h_{\ell_i,i}-h_{\ell_{i-1},i}), 
\end{align} 
and
\begin{align}
f_{(\ell_1,\ldots,\ell_n)}(h) &= f_{\ell_1,1}(h)+\sum_{i=2}^n (f_{\ell_i,i}(h)-f_{\ell_{i-1},i}(h)).
\end{align} 
Due to \eqref{eq:order2}, $h \leq_C \tilde h$ implies that $g_{(\ell_1,\ldots,\ell_n)}(h) \leq
g_{(\ell_1,\ldots,\ell_n)}(\tilde h)$.

\begin{assump}\label{ass:fl}
The functions $f_{(\ell_1,\ldots,\ell_n)}(h):\Omega \rightarrow \mathbb{R}$ are such that
 \[f_{(\ell_1,\ldots,\ell_n)}(h) \leq f_{(\ell_1,\ldots,\ell_n)}(\tilde  h)\] 
for any $h, \tilde h \in \Omega$
such that $h \leq_C \tilde h$ and $g_{(\ell_1,\ldots,\ell_n)}(\tilde h) = g_{(\ell_1,\ldots,\ell_n)}(h)$.
\end{assump}


The next assumption is used to prove that for any $h \leq_C \pi$ or $\pi \leq_C h$ the 
trajectory starting in $h$ of the set of ODEs converges to the fixed point $\pi$.

\begin{assump}\label{ass:sumfL1}
The functions $f_{\ell,1}(h)$ are such that for any fixed point $\pi$
and $L \geq 1$ we have
\begin{align*}
\sum_{\ell \geq L} (f_{\ell,1}(h)-f_{\ell,1}(\pi)) = 
\sum_{\ell =1}^{L-1} \sum_{j=1}^n b_{L,\ell,j}(h) (h_{\ell,j}-\pi_{\ell,j})
- a_{L,\pi}(h), 
\end{align*}
for some bounded functions $b_{L,\ell,j}(h)$ on $\Omega$ and functions $a_{L,\pi}(h)$ 
for which $a_{L,\pi}(h) \geq 0$ if $\pi \leq_C h$ and
$a_{L,\pi}(h) \leq 0$ if $h \leq_C \pi$.     
\end{assump}

The main theorem is stated below.

\begin{theo}[Global attraction]\label{th:global}
Consider the set of ODEs given by (\ref{eq:newODE1}-\ref{eq:newODE2}). 
Assume that (\ref{ass:unique}-\ref{ass:sumfL1}) hold and $\mu_i(1-p_i)$ is decreasing in $i$. 
Then, for any $h(0) \in \Omega$, $h(t)$ converges pointwise to the unique fixed point $\pi \in \Omega$ as $t$ tends to infinity.
\end{theo}
\begin{proof}
In Section \ref{sec:proof} we show that any fixed point $\pi \in \Omega$ 
is a global attractor, which implies that
the fixed point is unique.
\end{proof}

\begin{corol}\label{cor:Kurtz}
Consider a density dependent population process as defined by
Kurtz \cite{kurtz1} on $\Omega_B$ such that its drift given
by the right hand side of (\ref{eq:newODE1}-\ref{eq:newODE2})
is Lipschitz continuous. 
Let $X^{(N)}_\infty$ be the stationary measure of the $N$-th population
process and assume (\ref{ass:unique}-\ref{ass:sumfL1}) hold, 
then $X^{(N)}_\infty$ converges weakly to the Dirac
measure on $\pi$.
\end{corol}
\begin{proof}

As the $N$-th population process has a finite number of states it has
a unique stationary measure $X^{(N)}_\infty$. This sequence of measures
is tight as $\Omega_B$ is compact, thus any subsequence has a further subsequence
that converges to some limit point. Due to 
Theorem 3.5 and Corollary 3.9 in \cite{roth2013stochastic}
any such limit point has support on the Birkhoff center of the set of ODEs\footnote{Note that
this result holds in a  more general setting that the one
considered here, where the drift is not necessarily Lipschitz continuous
and is characterized by a differential inclusion.}.
As $\pi$ is a global attractor, the Birkhoff center is the singleton $\{\pi\}$
and the only possible limit point is therefore the Dirac measure on $\pi$.
Thus, every subsequence of $X^{(N)}_\infty$ has a further subsequence that
converges to the same limit, which implies that the entire sequence converges
to this limit. 
\end{proof}

Note that the above corollary is very general. In order to apply it, we do
need to truncate the buffer to some finite size $B$ (as in \cite{gast2010mean,ganesh2010}). This is 
not a real restriction from a practical point of view as there is 
virtually no difference between having an infinite buffer or a huge finite buffer, say of
size $B=10^{15}$ (provided that the system is stable in case of an infinite
buffer). Establishing a similar result for infinite buffers is technically more demanding
as one needs to establish the existence of $X^{(N)}_\infty$ and prove
that this sequence converges. 

An issue regarding the convergence is that
$\Omega$ is not compact due to the condition $\sum_\ell h_{\ell,1} < \infty$
(note that as $\Omega$ is not a finite 
dimensional Euclidean space, compactness of a set depends on the norm used).  For systems
with exponential job sizes the following approach is often used, see \cite{arpanStSyst,martin1999,vvedenskaya3}.
One first drops the condition that prevents $\Omega$ from being compact, thus in
our case we consider $\Omega_\infty$ which equals $\Omega$ 
without the requirement $\sum_\ell h_{\ell,1} < \infty$.
Then one picks a suitable norm, for instance $||h-\tilde h|| = \sum_{i=1}^n \sum_\ell \frac{|h_{\ell,i}-\tilde h_{\ell,i}|^2}{2^\ell}$ in our case\footnote{Note that pointwise convergence of $h(t)$ to $\pi$ also implies convergence under this norm.}, such that $\Omega_\infty$ is compact.
Thus, by Prokhorov's theorem any subsequence $X^{(N_k)}_\infty$  of $X^{(N)}_\infty$ has a further subsequence $X^{(N_k')}_\infty$ that converges
to some measure on $\Omega_\infty$.
Next one argues that any such limit point $\pi^*$ necessarily concentrates  on $\Omega$.
Note that while $X^{(N)}_\infty(\Omega)=1$ for all $N$, weak convergence
does not immediately imply that $\pi^*(\Omega)=1$ as $\Omega$ is an open set. 

To show that $\pi^*(\Omega)=1$, it suffices that $E_{\pi^*}[\sum_\ell h_{\ell,1}] < \infty$.
As $E_{\pi^*}[\sum_\ell h_{\ell,1}] \leq \liminf E_{X^{(N_k')}_\infty}[\sum_\ell h_{\ell,1}]$
(due to Portmanteau's theorem as $\sum_\ell h_{\ell,1}$ is continuous and bounded from below), 
$\pi^*(\Omega)=1$ if $E_{X^{(N)}_\infty}[\sum_\ell h_{\ell,1}]$ is bounded by some constant $c$ for all $N$. Finally, this constant $c$ is shown to be the mean queue length in some finite stable queueing system (an M/M/1 queue in \cite{vvedenskaya3}, a set of $J$ queues with random routing in \cite{arpanStSyst} and a classic Jackson network in \cite{martin1999}).

Having established that $\pi^*(\Omega)=1$ for any limit point $\pi^*$, one can
use Theorem 1 of \cite{leboudec3} and the global attraction 
to show that $\pi^*$ is the Dirac measure $\delta_\pi$.
To apply this theorem weak convergence over finite time scales suffices.

In case of phase-type service exactly the same line of reasoning can be applied. The main
step that requires extra care is to show that the mean queue length of a 
queue in the $N$-th system is bounded by
some constant $c$, for instance by letting $c$ be
the mean queue length of a queue in a set of $N$ independent $M/PH/1$ queues
(which should hold for any load balancing strategy that performs better than random). 
Recall that an $M/PH/1$ queue (with a load below one) has a finite mean queue length
as the second moment of a phase-type distribution is finite.

\section{Examples revisited}\label{sec:examples2}

In this section we discuss assumptions \ref{ass:unique} to \ref{ass:sumfL1} for the
examples listed in Section \ref{sec:examples}. With respect to assumption \ref{ass:unique},
we only briefly discuss the existence of a fixed point as the existence of a unique
solution $h(t)$ with $h(0)=h$ for $h \in \Omega$ can be easily verified by
checking the Lipschitz continuity of the drift on $\Omega$.

\subsection{JSQ(d): Join-the-Shortest-Queue among d randomly selected servers}\label{ex:revisit2}

The existence of a fixed point when the buffer size $B$ is finite is easy to establish (see Section \ref{sec:JSQKd} with $K=1$).
For an infinite buffer size $B$, the existence of a fixed point in $\Omega$ follows from \cite[Section 8]{bramsonLB_QUESTA}
as the distributions belonging to $\mathcal{C}_0$ have a decreasing hazard rate. 

Assumption \ref{ass:f} is trivial to verify. 
To check whether Assumption \ref{ass:fl} holds, 
we can write $g_{(\ell_1,\ldots,\ell_n)}(h)$ as $\sum_{i=1}^n (h_{\ell_i,i}-h_{\ell_i,i+1})$  and similarly
$f_{(\ell_1,\ldots,\ell_n)}(h)$ equals $\sum_{i=1}^n (f_{\ell_i,i}(h)-f_{\ell_i,i+1}(h))$. Therefore,
\begin{align}
&f_{(\ell_1,\ldots,\ell_n)}(h) = \lambda 
\sum_{i=1}^n 1[\ell_i > 1]  \left(  \sum_{j=0}^{d-1} h_{\ell_i-1,1}^j h_{\ell_i,1}^{d-1-j} \right)
\nonumber \\
& \hspace*{1.5cm} \cdot [(h_{\ell_i-1,i}-h_{\ell_i,i})-(h_{\ell_i-1,i+1}-h_{\ell_i,i+1})] \nonumber \\
 &\ \ \ \ = \lambda 
\sum_{i=1}^n 1[\ell_i > 1]  \left(  \sum_{j=0}^{d-1} h_{\ell_i-1,1}^j h_{\ell_i,1}^{d-1-j} \right) 
\nonumber \\
& \hspace*{1.5cm} \cdot [(h_{\ell_i-1,i}-h_{\ell_i-1,i+1})-(h_{\ell_i,i}-h_{\ell_i,i+1})] 
\label{eq:fls_SQd}
\end{align}
If $\ell_1 > \ell_2 > \ldots > \ell_n$, this can be written as 
\begin{align*}
&f_{(\ell_1,\ldots,\ell_n)}(h) = \lambda 
\sum_{i=1}^n \left(  \sum_{j=0}^{d-1} h_{\ell_i-1,1}^j h_{\ell_i,1}^{d-1-j} \right)\\
& \hspace*{1.5cm} \cdot [g_{(\ell_1,\ldots,\ell_{i-1},\ell_i-1[\ell_i > 1],\ell_{i+1},\ldots,\ell_n)}(h)-g_{(\ell_1,\ldots,\ell_n)}(h)]. 
\end{align*}
and Assumption \ref{ass:fl} holds as $h \leq_C \tilde h$ implies that $h_{\ell,i} \leq \tilde h_{\ell,i}$ and
$g_{(\ell_1',\ldots,\ell_n')}(h) \leq g_{(\ell_1',\ldots,\ell_n')}(\tilde h)$ for $(\ell_1',\ldots,\ell_n') \not= 
(\ell_1,\ldots,\ell_n)$.  

If $\ell_i = \ell_{i+1} > 1$ for some $i$, the above expression cannot be directly used
as $\ell_i - 1[\ell_i > 1] \geq \ell_{i+1}$ does not hold. In general
assume $\tilde \ell_1 > \tilde \ell_2 > \ldots > \tilde \ell_k$ are the unique values appearing in the 
sequence $\ell_1 \geq \ell_2 \geq \ldots \geq \ell_n$ and let $\ell_{j_i}$ be the
first element in this sequence equal to $\tilde \ell_i$, then \eqref{eq:fls_SQd} becomes
  \begin{align*}
&f_{(\ell_1,\ldots,\ell_n)}(h) = \lambda 
\sum_{i=1}^k 1[\tilde \ell_i > 1]  \left(  \sum_{j=0}^{d-1} h_{\tilde \ell_i-1,1}^j h_{\tilde \ell_i,1}^{d-1-j} \right)\\
& \hspace*{1.5cm} \cdot [(h_{\tilde \ell_i-1,j_i}-h_{\tilde \ell_i-1,j_{i+1}})-
(h_{\tilde \ell_i,j_i}-h_{\tilde \ell_i,j_{i+1}})]. 
\end{align*}
Assumption \ref{ass:fl} now follows as $(h_{\tilde \ell_i-1,j_i}-h_{\tilde \ell_i-1,j_{i+1}})-
(h_{\tilde \ell_i,j_i}-h_{\tilde \ell_i,j_{i+1}})$ can be written
as $g_{(\ell_1',\ldots,\ell_n')}(h)-g_{(\ell_1,\ldots,\ell_n)}(h)$
with $\ell_s' = \ell_s - 1$ for $j_i \leq s < j_{i+1}$ and
$\ell_s'=\ell_s$ otherwise. Note that $\ell_1' \geq \ell_2' \geq \ldots \geq \ell_n'$ as
required because $\tilde \ell_i > \tilde \ell_{i+1}$.

Assumption \ref{ass:sumfL1} with $L = 1$ is immediate as $\sum_{\ell \geq 1} f_{\ell,1}(h) = \lambda$
for any $h \in \Omega$, meaning we can pick $a_{1,\pi}(h)=0$.
Finally, as $\sum_{\ell \geq L} (f_{\ell,1}(h)-f_{\ell,1}(\pi)) = \lambda (h_{L-1,1}^d - 
\pi_{L-1,1}^d)$, setting
$a_{L,\pi}(h) = 0$, $b_{L,L-1,1}(h)=\lambda \sum_{j=0}^{d-1} (h_{L-1,1})^j (\pi_{L-1,1})^{d-1-j} \leq \lambda d$ 
and $b_{L,\ell,j}(h)=0$ for $(\ell,j) \not= (L-1,1)$ verifies 
Assumption \ref{ass:sumfL1} with $L > 1$.

When the buffer is finite of size $B$, the above discussion remains valid, except that we
need to set $a_{L,\pi}(h) = \lambda (h_{B,1}^d-\pi_{B,1}^d)$, for $L \geq 1$, such that Assumption \ref{ass:sumfL1} holds.

\subsection{Pull and push strategies}

In this example it is possible to show that for $\lambda < 1$ the set of ODEs has a unique fixed point that can be computed by determining the invariant distribution of an ergodic Quasi-Birth-Death Markov chain
\cite{vanhoudt_stealing}. Note that the issue of global attraction is not
addressed in \cite{vanhoudt_stealing}.
Assumption \ref{ass:f} is readily verified.
To verify Assumption \ref{ass:fl} it is not hard to show
that $f_{(\ell_1,\ldots,\ell_n)}(h)$ can be written as
\begin{align*}
 f_{(\ell_1,\ldots,\ell_n)}&(h) = \lambda  g_{(\ell_1-1,\ell_2-1[\ell_2>1],\ldots,\ell_n-1[\ell_n>1])}(h)\\&-\lambda g_{(\ell_1,\ldots,\ell_n)}(h)-r(1-h_{1,1}) g_{(\ell_1,\ldots,\ell_n)}(h)\\
& +r(1-h_{1,1}) g_{(\ell_1+1,\ell_2+1[\ell_2 > 1],\ldots,\ell_n+1[\ell_n > 1])}(h).
\end{align*}
Thus if $g_{(\ell_1,\ldots,\ell_n)}(h) = g_{(\ell_1,\ldots,\ell_n)}(\tilde h)$, then
\begin{align*}
f_{(\ell_1,\ldots,\ell_n)}&(\tilde h)- f_{(\ell_1,\ldots,\ell_n)}(h) = \lambda  g_{(\ell_1-1,\ell_2-1[\ell_2>1],\ldots,\ell_n-1[\ell_n>1])}(\tilde h)\\
&-\lambda g_{(\ell_1-1,\ell_2-1[\ell_2>1],\ldots,\ell_n-1[\ell_n>1])}(h)\\
&+r(1-\tilde h_{1,1}) g_{(\ell_1+1,\ell_2+1[\ell_2 > 1],\ldots,\ell_n+1[\ell_n > 1])}(\tilde h)\\
&- r(1-h_{1,1}) g_{(\ell_1+1,\ell_2+1[\ell_2 > 1],\ldots,\ell_n+1[\ell_n > 1])}(h) \\
&+ r (\tilde h_{1,1}-h_{1,1}) g_{(\ell_1,\ldots,\ell_n)}(\tilde h) \geq 0
\end{align*}
if $h \leq_C \tilde h$ as $g_{(\ell_1,\ldots,\ell_n)}(\tilde h) \geq g_{(\ell_1+1,\ell_2+1[\ell_2 > 1],\ldots,\ell_n+1[\ell_n > 1])}(\tilde h)$.

Assumption \ref{ass:sumfL1} with $L = 1$ follows by noting that $\sum_{\ell \geq 1} f_{\ell,1}(h) = \lambda = \sum_{\ell \geq 1} f_{\ell,1}(\pi)$.
Finally, Assumption \ref{ass:sumfL1} with $L > 1$ can be verified as follows.
Note that 
\begin{align*}
\sum_{\ell \geq L} &(f_{\ell,1}(h)-f_{\ell,1}(\pi)) = \lambda (h_{L-1,1}-\pi_{L-1,1})\\
& - r((1-h_{1,1}) h_{L,1} - (1-\pi_{1,1}) \pi_{L,1}) =\lambda (h_{L-1,1}-\pi_{L-1,1})  \\
&+ r(h_{1,1}-\pi_{1,1})h_{L,1} - r(1-\pi_{1,1})(h_{L,1}-\pi_{L,1}),  
\end{align*}
meaning Assumption \ref{ass:sumfL1} with $L > 1$ holds with $a_{L,\pi}(h) = r(1-\pi_{1,1})(h_{L,1}-\pi_{L,1})$, $b_{L,1,1}(h)= rh_{L,1} \leq r$,
$b_{L,L-1,1}(h) = \lambda$ and all other $b_{L,\ell,j}(h)$ equal to zero.

\subsection{JSQ(K,d): Join-the-Shortest-K-Queues among d randomly selected servers } \label{sec:JSQKd}

With respect to assumption \ref{ass:unique}, we limit ourselves to the case where $B$ is finite. 
The existence of a fixed point in $\Omega_B$ follows from the fact that a convex compact forward invariant set $\mathcal{K} \subset \mathbb{R}^{Bn}$
of a dynamical system has a fixed point in $\mathcal{K}$ \cite{bhatia2002stability}, which is not hard to prove using Brouwer's fixed point
theorem.

Contrary to the previous two examples, verifying Assumption \ref{ass:f} requires some work. First note that 
$f_{\ell,1}(h)$ only depends on $h_{\ell-1,1}$ and $h_{\ell,1}$ and therefore the functions $f_{\ell,1}(h)$
are non-decreasing in $h_{\ell',i'}$ for $(\ell',i')\not=(\ell,1)$ if
\[\phi_K(x) = \sum_{s=0}^{K-1} (K-s) {d \choose s} x^{d-s} (1-x)^s, \]
is non-decreasing on $[0,1]$. We now prove that 
\begin{align}\label{eq:fpx}
\phi_K'(x) = \sum_{s=0}^{K-1} d {d-1 \choose s} x^{d-s-1} (1-x)^s, 
\end{align}
which is clearly positive on $[0,1]$.
By definition of $\phi_K(x)$ we have
\[\phi_K'(x) = \sum_{s=0}^{K-1} (K-s) {d \choose s} x^{d-s-1} (1-x)^{s-1} ((1-x)d-s). \]
By induction on $K$ we find
\begin{align*}\phi_K'(x) &= \sum_{s=0}^{K-1} {d \choose s} x^{d-s-1} (1-x)^{s-1} ((1-x)d-s) \\&+
\sum_{s=0}^{K-2} d {d-1 \choose s} x^{d-s-1} (1-x)^s.\end{align*} 
Hence, \eqref{eq:fpx} is equivalent to showing that
\begin{align*}
\sum_{s=0}^{K-1}  &d{d \choose s}  x^{d-s-1} (1-x)^s = \\
 & \sum_{s=1}^{K-1}  d{d-1 \choose s-1} x^{d-s-1} (1-x)^{s-1} + d{d-1 \choose K-1} x^{d-K} (1-x)^{K-1}.
\end{align*} 
which is easy to establish (using induction on $K$ once more).

Let us now focus on the functions $f_{\ell,i}(h)$ with $i > 1$. Clearly, these functions are increasing in
$h_{\ell-1,i}$. It remains to show that they are also increasing in $h_{\ell-1,1}$ and $h_{\ell,1}$, which holds if
\[ \xi_K(x_1,x_2) = \frac{\phi_K(x_2)-\phi_K(x_1)}{x_2-x_1},\]
is increasing in both components for $0 \leq x_1 \leq x_2 \leq 1$. As $\xi_K(x_1,x_2)$ is symmetric, it suffices to argue that $\xi_K(x_1,x_2)$
is increasing in $x_1$. Further, demanding that the derivative of $\xi_K(x_1,x_2)$ with respect to $x_1$ is non-negative
is equivalent to
\[ \phi_K(x_2) \geq \phi_K(x_1) + \phi_K'(x_1) (x_2-x_1),\]
which holds if and only if $\phi_K(x)$ is convex.
Using \eqref{eq:fpx} we have for $K < d$
\[\phi_K''(x) = \sum_{s=0}^{K-1} d {d-1 \choose s} x^{d-s-2} (1-x)^s ((d-1)(1-x)-s).\]
Using induction on $K$ this can be rewritten as
\[\phi_K''(x) = d(d-1) {d-2 \choose K-1} x^{d-K-1} (1-x)^{K-1},\]
which is clearly positive on $[0,1]$. For $K=d$, we have $\phi_K''(x)=0$ as $\phi_K(x)=dx$.

We now proceed with Assumption \ref{ass:fl}.
As $f_{(\ell_1,\ldots,\ell_n)}(h)$ equals $\sum_{i=1}^n (f_{\ell_i,i}(h)-f_{\ell_i,i+1}(h))$ and $\ell_1 > 1$, 
we note that
\begin{align*}
 f_{(\ell_1,\ldots,\ell_n)}(h)  =\lambda 
& \sum_{i=1}^n 1[\ell_i > 1]  \xi_K(h_{\ell_i,1},h_{\ell_i-1,1}) \\
& \ \ \cdot [(h_{\ell_i-1,i}-h_{\ell_i,i})-(h_{\ell_i-1,i+1}-h_{\ell_i,i+1})]
\end{align*}
If $\ell_1 > \ell_2 > \ldots > \ell_n$, this can be written as 
\begin{align*}
 & f_{(\ell_1,\ldots,\ell_n)}(h)  =\lambda  
\sum_{i=1}^n \xi_K(h_{\ell_i,1},h_{\ell_i-1,1}) 
\\
& \ \ \cdot [g_{(\ell_1,\ldots,\ell_{i-1},\ell_i-1[\ell_i > 1],\ell_{i+1},\ldots,\ell_n)}(h)
-g_{(\ell_1,\ldots,\ell_n)}(h)], 
\end{align*}
where $\xi_K(x_1,x_2)$ was increasing in $x_1$ and $x_2$. 
The case where $\ell_i = \ell_{i+1}$ for some $i$ can be dealt
with in the same manner as in Example \ref{ex:revisit2}.

Assumption \ref{ass:sumfL1} with $L = 1$  holds as $\sum_{\ell \geq 1} f_{\ell,1}(h) = \lambda K = \sum_{\ell \geq 1} f_{\ell,1}(\pi)$.
Finally, with respect to Assumption \ref{ass:sumfL1}  with $L > 1$ we have
\begin{align*}
\sum_{\ell \geq L} (f_{\ell,1}(h)-f_{\ell,1}(\pi)) &= \lambda (\phi_K(h_{L-1,1}) - \phi_K(\pi_{L-1,1})) \\
 &= \lambda \xi_K(\pi_{L-1,1},h_{L-1,1}) (h_{L-1,1}-\pi_{L-1,1}),
\end{align*}
and $\xi_K(x_1,x_2) \leq \phi_K'(1) = d$ due to the convexity of $\phi_K(x)$. 

\section{Proof of Theorem 2}\label{sec:proof}

In this section we define $\nu_i = \mu_i (1-p_i)$ to ease the notation.
We start by showing that the order $\leq_C$ is preserved over time.

\begin{prop}\label{prop:mono}
Assume that (\ref{ass:unique}-\ref{ass:fl}) hold and let $h,\tilde h \in \Omega$.  
Let $h(t)$ and $\tilde h(t)$ be the unique solution of (\ref{eq:newODE1}-\ref{eq:newODE2}) with $h(0) = h$
and $\tilde h(0) = \tilde h$, respectively.
If $\mu_i(1-p_i)$ is decreasing in $i$ and 
$h \leq_C \tilde h$, then $h(t) \leq_C \tilde h(t)$ for any $t > 0$.
\end{prop}
\begin{proof}
Assume that at some time $t$ we have $h_{\ell,i}(t) = \tilde h_{\ell,i}(t)$ for some $\ell$ and $i$, while
$h(t) \leq_C  \tilde h(t)$. We need to argue that 
\[ d\tilde h_{\ell,i}(t)/dt \geq  dh_{\ell,i}(t)/dt,\] 
as the order is otherwise violated at time $t+$. 

As
$h(t) \leq_C  \tilde h(t)$ implies that $h_{\ell',i'}(t) \leq \tilde h_{\ell',i'}(t)$
for all $\ell'$ and $i'$, it would be sufficient that $dh_{\ell,i}(t)/dt$ is non-decreasing in all $h_{\ell',i'}(t)$
with $(\ell',i') \not= (\ell,i)$.  Looking at (\ref{eq:newODE1}-\ref{eq:newODE2}) and due to
Assumption \ref{ass:f},  we see that this is clearly the case, except perhaps for the sums over $j$ (that are due to the service completions).

For $i > 1$, we have
\begin{align*}
-\sum_{j=i}^n & (h_{\ell,j}(t)-h_{\ell,j+1}(t)) \nu_j = \sum_{j=i+1}^n h_{\ell,j}(t) (\nu_{j-1}-\nu_{j}) -h_{\ell,i}(t) \nu_i,
\end{align*} 
meaning $dh_{\ell,i}(t)/dt$ is non-decreasing in any $h_{\ell',i'}(t)$ with $(\ell',i') \not= (\ell,i)$  when
$i > 1$, as $\nu_i = \mu_i(1-p_i)$ is decreasing in $i$ (and positive).

For $i=1$, we find
\begin{align*}
 -\sum_{j=1}^n &\left[(h_{\ell,j}(t)-h_{\ell,j+1}(t))-(h_{\ell+1,j}(t)-h_{\ell+1,j+1}(t))\right]\nu_j  \\
&= \sum_{j=2}^n (h_{\ell,j}(t)-h_{\ell+1,j}(t)) (\nu_{j-1}-\nu_j) \\
&\ \ \ \ \ \ -(h_{\ell,1}(t)-h_{\ell+1,1}(t)) \nu_1 \\
&= \sum_{j=2}^n (h_{\ell+1,1}(t)+ h_{\ell,j}(t)-h_{\ell+1,j}(t)) (\nu_{j-1}-\nu_j) \\
&\ \ \ \ \ \ + h_{\ell+1,1}(t) \nu_n  - h_{\ell,1}(t) \nu_1. 
\end{align*} 
This expression is decreasing in
 $h_{\ell+1,j}(t)$, for $j > 1$, which may appear as a problem. 
However, as $h(t) \leq_C \tilde h(t)$, 
\eqref{eq:order2} with $\ell_1 = \ldots = \ell_{j-1}=\ell+1$ and $\ell_j = \ldots = \ell_n = \ell$
implies that $h_{\ell,j}(t)+h_{\ell+1,1}(t)-h_{\ell+1,j}(t) \leq 
\tilde h_{\ell,j}(t)+\tilde h(t)_{\ell+1,1}-\tilde h_{\ell+1,j}(t)$.
As a result $dh_{\ell,1}(t)/dt$ does not exceed $d\tilde h_{\ell,1}(t)/dt$ as required.

We also need to verify that \eqref{eq:order2} remains valid, which corresponds to verifying that
$g_{(\ell_1,\ldots,\ell_n)}(h(t)) \leq g_{(\ell_1,\ldots,\ell_n)}(\tilde h(t))$ remains valid. 
Assume that $g_{(\ell_1,\ldots,\ell_n)}(h(t)) = g_{(\ell_1,\ldots,\ell_n)}(\tilde h(t))$
for some $(\ell_1,\ldots,\ell_n)$, then we need to argue that
\[dg_{(\ell_1,\ldots,\ell_n)}(\tilde  h(t))/dt \geq  dg_{(\ell_1,\ldots,\ell_n)}(h(t))/dt,\] 
whenever $h(t) \leq_C \tilde h(t)$ to complete the proof. Due to Assumption \ref{ass:fl}, we can restrict ourselves
to showing that the terms of $dg_{(\ell_1,\ldots,\ell_n)}(h(t))/dt$ corresponding to phase changes and service completions   are increasing in
$g_{(\ell_1',\ldots,\ell_n')}(h(t))$ when $(\ell_1',\ldots,\ell_n')\not=(\ell_1,\ldots,\ell_n)$.

Phase changes increase $g_{(\ell_1,\ldots,\ell_n)}(h(t))$ if such a change occurs in a queue
in phase $i$ with a length in $[\ell_{i+1},\ell_i -1]$. Therefore  phase changes increase
$g_{(\ell_1,\ldots,\ell_n)}(t)$  at rate
\[ \sum_{i=1}^{n-1} \mu_i p_i (g_{(\ell_1,\ldots,\ell_{i-1},\ell_{i+1},\ell_{i+1},\ldots,\ell_n)}(h(t))-g_{(\ell_1,\ldots,\ell_n)}(h(t))).\]
Service completions in queues in phase $i$ that have a length in $[l_i,l_1]$ decrease $g_{(\ell_1,\ldots,\ell_n)}(h(t))$
at rate $\nu_i$ (as the initial service phase of the next job in service is phase $1$). The drift is therefore given by
\begin{align*}
-\sum_{i=1}^n &\nu_i (g_{(\underbrace{\scriptstyle \ell_1+1,\ldots,\ell_1+1}_{i-1 \ \mbox{\tiny   times}},\ell_i,\ldots,\ell_n)}(h(t))
- g_{(\underbrace{\scriptstyle \ell_1+1,\ldots,\ell_1+1}_{i \ \mbox{\tiny times}},\ell_{i+1},\ldots,\ell_n)}(h(t)))\\
&=  \sum_{i=2}^n g_{(\underbrace{\scriptstyle \ell_1+1,\ldots,\ell_1+1}_{i-1 \ \mbox{\tiny times}},\ell_i,\ldots,\ell_n)}(h(t))
(\nu_{i-1} -\nu_i) \\ 
&-g_{(\ell_1,\ldots,\ell_n)}(h(t)) \nu_1  + g_{(\ell_1+1,\ldots,\ell_1+1)}(h(t)) \nu_n.
\end{align*}  
\end{proof}

The next proposition shows that it suffices to prove attraction for points $h \in \Omega$
for which $h \leq_C \pi$ or $\pi \leq_C h$, where $\pi$ is a fixed point of the ODEs (\ref{eq:newODE1}-\ref{eq:newODE2}).

\begin{prop}\label{prop:sufficient}
Let $h \in \Omega$ and assume $\mu_i(1-p_i)$ is decreasing in $i$, then the trajectory $h(t)$ starting in $h(0) \in \Omega$ 
at time $0$ converges pointwise to $\pi$ provided that
for any $h \in \Omega$ with $h \leq_C \pi$ or $\pi \leq_C h$, $h(t)$ with $h(0)=h$ converges pointwise to $\pi$.
\end{prop}
\begin{proof}
Due to Proposition \ref{prop:mono} it suffices to show that for any $h \in \Omega$ there exists
a $h^{(l)}, h^{(u)} \in \Omega$, with $h^{(l)} \leq_C \pi$ and $\pi \leq_C h^{(u)}$, 
such that $h^{(l)} \leq_C h \leq_C h^{(u)}$.
For $h^{(l)}$ we can simply pick the zero vector as $0 \leq_C h$ for any $h \in \Omega$.
For $h^{(u)}$ we set $h_{\ell,i}^{(u)} = \max(h_{\ell,1},\pi_{\ell,1})$, for $i=1,\ldots,n$. Hence,
\begin{align*}
h_{\ell,i} &\leq h_{\ell,1} \leq \max(h_{\ell,1},\pi_{\ell,1}) = h_{\ell,i}^{(u)} 
\end{align*} 
and for $\ell_1 \geq \ell_2 \geq \ldots \geq \ell_n \geq 1$ with $\ell_1 > \ell_n$
\begin{align*}
 h_{\ell_1,1}+&\sum_{i=2}^n (h_{\ell_i,i}-h_{\ell_{i-1},i}) \leq h_{\ell_n,1} \leq \max(h_{\ell_n,1},\pi_{\ell_n,1}) = h^{(u)}_{\ell_n,1} 
\\ & = h^{(u)}_{\ell_1,1} + \sum_{i=2}^n (h^{(u)}_{\ell_i,1}-h^{(u)}_{\ell_{i-1},1}) = h^{(u)}_{\ell_1,1} + \sum_{i=2}^n (h^{(u)}_{\ell_i,i}-h^{(u)}_{\ell_{i-1},i}), 
\end{align*} 
where the first inequality follows from \eqref{eq:helpbound}.
This shows that $h \leq_C h^{(u)}$ and similarly one finds that $\pi \leq_C h^{(u)}$. 
\end{proof}

Remark that when $B$ is finite we can simply use $h^{(u)}_{B,n} = 1$ and $h^{(u)}_{\ell,i} = 0$
for $(\ell,i)\not= (B,n)$.

\begin{lem}\label{lem:zL1z2}
Define $z_{1,L}(h(t)) = \sum_{\ell \geq L} h_{\ell,1}(t)$ and 
\[z_2(h(t))=\sum_{i=2}^n h_{1,i}(t) (R_i - R_{i-1}).\]
Then, 
\begin{align}\label{eq:z1L}
\frac{d}{dt}z_{1,L}(h(t)) &= \sum_{\ell \geq L} f_{\ell,1}(h(t)) - \sum_{j=1}^n (h_{L,j}(t)-h_{L,j+1}(t))\nu_j, 
\end{align}
and
\begin{align}\label{eq:z2}
\frac{d}{dt}z_2(h(t)) &= -h_{1,1}(t) + \sum_{j=1}^n (h_{1,j}(t)-h_{1,j+1}(t))\nu_j.
\end{align}
\end{lem}
\begin{proof}
The expression for $dz_{1,L}(h(t))/dt$ is immediate from \eqref{eq:newODE1}. For $dz_2(h(t))/dt$ we can make use of
\eqref{eq:newODE2} and obtain (after exchanging the order of the sums):
\begin{align*}
\frac{d}{dt}z_2(h(t)) &= \sum_{i=2}^n (h_{1,i-1}(t)-h_{1,i}(t))p_{i-1}\mu_{i-1}R_i \\
&-\sum_{i=2}^n (h_{1,i-1}(t)-h_{1,i}(t))p_{i-1}\mu_{i-1}R_{i-1} \\
&-\sum_{j=2}^n \left(\sum_{i=2}^j (R_i - R_{i-1}) \right) (h_{1,j}(t)-h_{1,j+1}(t)) \nu_j.
\end{align*}
Using \eqref{eq:Ris} on the first sum, we can rewrite this as
\begin{align*}
\frac{d}{dt}&z_2(h(t)) = \sum_{i=2}^n (h_{1,i-1}(t)-h_{1,i}(t))\nu_{i-1}R_{i-1} \\
&- \sum_{i=2}^n (h_{1,i-1}(t)-h_{1,i}(t))-\sum_{j=2}^n  (h_{1,j}(t)-h_{1,j+1}(t)) \nu_j R_j \\
&+ \sum_{j=2}^n  (h_{1,j}(t)-h_{1,j+1}(t)) \nu_j R_1\\
&\hspace*{-3mm} = (h_{1,1}(t)-h_{1,2}(t))\nu_1 R_1 - h_{1,1}(t)+h_{1,n}(t)-h_{1,n}(t)\nu_n R_n \\
&+ \sum_{j=2}^n  (h_{1,j}(t)-h_{1,j+1}(t)) \nu_j R_1.
\end{align*}
The result follows by noting that $R_1 = 1$ and $R_n = 1/\mu_n$.
\end{proof}

\begin{prop}\label{prop:attract}
Assume that (\ref{ass:unique}-\ref{ass:sumfL1}) hold and that $\mu_i(1-p_i)$ is decreasing in $i$. For any $h(0) \in \Omega$ with $h(0) \leq_C \pi$ or $\pi \leq_C h(0)$, $h(t)$ converges pointwise to $\pi$.
\end{prop}
\begin{proof}
We assume $\pi \leq_C h(0)$, the proof for $h(0) \leq_C \pi$ proceeds similarly.
We first show that $h_{1,1}(t)$ converges to $\pi_{1,1}$. As $\pi \leq_C h(t)$ for $t \geq 0$
due to Proposition \ref{prop:mono}, it suffices to show $\int_{t=0}^\infty (h_{1,1}(t) -\pi_{1,1}) dt < \infty$.
Let $z(h) = z_{1,1}(h) + z_2(h)$, then by Lemma \ref{lem:zL1z2} and Assumption \ref{ass:sumfL1} with $L = 1$, we have 
 $dz(h(t))/dt = \sum_{\ell \geq 1} f_{\ell,1}(\pi) - h_{1,1}(t) -a_1(h(t)) \leq 0$. Further $\sum_{\ell \geq 1} f_{\ell,1}(\pi) = \pi_{1,1}$
as $dz(\pi)/dt=0$. Therefore, 
\begin{align*}
\int_{t=0}^\tau (h_{1,1}(t) -\pi_{1,1}) dt &= -\int_{t=0}^\tau \frac{dz(t)}{dt} dt - \int_{t=0}^\tau a_1(h(t)) dt\\
&\leq z(h(0)) - z(h(\tau)) \leq z(h(0)),
\end{align*}
as $\pi \leq_C h(t)$ and  $z(h(\tau)) \geq 0$ for $\tau \geq 0$.
Hence, $\int_{t=0}^\tau (h_{1,1}(t) -\pi_{1,1}) dt$ is uniformly bounded in $\tau \geq 0$, meaning $\int_{t=0}^\infty (h_{1,1}(t) -\pi_{1,1}) dt < \infty$.
 
We now show that $h_{1,j}(t)$ converges to $\pi_{1,j}$, for $j=2,\ldots,n$, 
by arguing that $\int_{t=0}^\tau (h_{1,j}(t) -\pi_{1,j}) dt$ is uniformly bounded in $\tau \geq 0$. As $\nu_{j-1}-\nu_j > 0$ it suffices
to show that
\[\int_{t=0}^\tau \sum_{j=2}^n (h_{1,j}(t) -\pi_{1,j}) (\nu_{j-1}-\nu_j) dt\]
is uniformly bounded in $\tau \geq 0$. 
As $z_2(h(\tau)) \geq 0$, we have $z_2(h(0)) \geq -\int_{t=0}^\tau \frac{dz_2(h(t))}{dt} dt $ and 
Lemma \ref{lem:zL1z2}
implies
\begin{align*}
z_2(h(0)) &\geq  \int_{t=0}^\tau h_{1,1}(t) dt - \int_{t=0}^\tau \sum_{j=1}^n (h_{1,j}(t) -h_{1,j+1}(t)) \nu_j dt \\
&= \int_{t=0}^\tau h_{1,1}(t) (1-\nu_1) dt + \int_{t=0}^\tau \sum_{j=2}^n h_{1,j}(t) (\nu_{j-1}-\nu_j) dt \\
&= \int_{t=0}^\tau (h_{1,1}(t)-\pi_{1,1}) (1-\nu_1) dt \\
& \ \ \ \ + \int_{t=0}^\tau \sum_{j=2}^n (h_{1,j}(t)-\pi_{1,j}) (\nu_{j-1}-\nu_{j}) dt, 
\end{align*}
where the last inequality is due to the fact that $\pi$ is a fixed point. This shows the uniform boundedness 
in $\tau$ as $0 \leq \int_{t=0}^\tau (h_{1,1}(t)-\pi_{1,1}) dt \leq z(h(0))$.

We complete the proof by showing that $(h_{L,j}(t)-h_{L,j+1}(t))$ converges to $(\pi_{L,j}-\pi_{L,j+1})$, for $L > 1$ and $j=1,\ldots,n$.
Note that $(h_{L,j}(t)-h_{L,j+1}(t))$ is not necessarily larger than $(\pi_{L,j}-\pi_{L,j+1})$ when $\pi \leq_C h(t)$.
We do however have that $h_{L,j}(t) + h_{L-1,j+1}(t) - h_{L,j+1}(t) \geq  \pi_{L,j} + \pi_{L-1,j+1} - \pi_{L,j+1}$
when $\pi \leq_C h(t)$ due to \eqref{eq:order2}. Thus, using induction on $L$ it suffices to show that
\begin{align}\label{eq:toshowL}
\Psi_{L,\tau} = \int_{t=0}^\tau  \sum_{j=1}^n ((h_{L,j}(t) &+ h_{L-1,j+1}(t) - h_{L,j+1}(t)) \nonumber \\
& -(\pi_{L,j} + \pi_{L-1,j+1} - \pi_{L,j+1}))\nu_j dt
\end{align}
is uniformly bounded in $\tau \geq 0$. 
As $z_{1,L}(h(\tau)) \geq 0$ for $\tau \geq 0$, we have $z_{1,L}(h(0)) \geq -\int_{t=0}^\tau \frac{dz_{1,L}(h(t))}{dt} dt$ and Lemma \ref{lem:zL1z2}
for $L > 1$ implies
\begin{align*}
z_{L,1}&(h(0)) \geq \\
&-\int_{t=0}^\tau \sum_{\ell \geq L} f_{\ell,1}(h(t))dt  
+ \int_{t=0}^\tau \sum_{j=1}^n (h_{L,j}(t) -h_{L,j+1}(t)) \nu_j dt \\
&= -\int_{t=0}^\tau \sum_{\ell \geq L} (f_{\ell,1}(h(t))-f_{\ell,1}(\pi)) dt\\
& \ \ \ \ + \int_{t=0}^\tau \sum_{j=1}^n ((h_{L,j}(t) -h_{L,j+1}(t))-(\pi_{L,j}-\pi_{L,j+1})) \nu_j dt 
\end{align*}
where the last equality holds as $\pi$ is a fixed point. Therefore we find
\begin{align*}
z_{L,1}&(h(0)) + \int_{t=0}^\tau \sum_{j=1}^n (h_{L-1,j+1}(t)-\pi_{L-1,j+1}) \nu_j dt  \\
& \geq -\int_{t=0}^\tau \sum_{\ell \geq L} (f_{\ell,1}(h(t))-f_{\ell,1}(\pi)) dt + \Psi_{L,\tau} 
\end{align*}
By relying on Assumption \ref{ass:sumfL1} with $L > 1$ this can be restated as
\begin{align*}
z_{L,1}(h(0)) &+  \int_{t=0}^\tau \sum_{j=1}^n (h_{L-1,j+1}(t)-\pi_{L-1,j+1}) \nu_j dt \\
&+ \int_{t=0}^\tau \sum_{\ell =1}^{L-1} \sum_{j=1}^n  b_{L,\ell,j}(h(t)) (h_{\ell,j}(t)-\pi_{\ell,j}) dt \\ \geq 
& \Psi_{L,\tau}  +  \int_{t=0}^\tau a_{L,\pi}(h(t)) dt.
\end{align*}
As $b_{L,\ell,j}(h)$ is bounded on $\Omega$, the left hand side is uniformly bounded in $\tau$ by induction on $L$
and therefore so are the (positive) integrals on the right hand side. 
\end{proof}

Theorem \ref{th:global} follows by combining Proposition \ref{prop:sufficient} and \ref{prop:attract}.

\section{Conclusions}\label{sec:concl}

In this paper we demonstrated that monotonicity arguments can still be applied to prove global
attraction of mean field models with hyperexponential job sizes, which is a widely used class of distributions
for systems exhibiting large job size variability. The key ideas to enable the use of such monotonicity
arguments existed in formulating the ODE-based mean field model using a Coxian representation 
and introducing a partial order that is stronger than the usual componentwise order.

We believe that the approach presented in this paper can be extended to heterogeneous systems
and systems in which a server can serve multiple jobs simultaneously (i.e., in which each server
behaves as an $\cdot/Cox/C$ server). Whether the assumption
on the first-come-first-served scheduling discipline can be relaxed is unclear at this
moment and is the topic of future work.

\bibliographystyle{plain}
\bibliography{../PhD/thesis}

\appendix

\section{Proof of Lemma 2}\label{app:lemRi}

By definition of $R_i$ we have
\[R_i = \frac{1}{\mu_i} + \sum_{j=i}^{n-1} \left( \prod_{k=i}^j p_k \right) \frac{1}{\mu_{j+1}},\]
which implies that
\[R_i - R_{i-1} = \frac{1-p_{i-1}}{\mu_i} + \sum_{j=i}^{n-1} \left( \prod_{k=i}^j p_k \right) \frac{1-p_{i-1}}{\mu_{j+1}} - \frac{1}{\mu_{i-1}}.\]
As $\mu_i(1-p_i)$ is decreasing in $i$, we have $(1-p_{i-1})/\mu_{j+1} > (1-p_{j+1})/\mu_{i-1}$ for $j \geq i-1$.
Hence,
\begin{align*}
R_i - R_{i-1} &> \frac{1-p_{i}}{\mu_{i-1}} + \sum_{j=i}^{n-1} \left( \prod_{k=i}^j p_k \right) \frac{1-p_{j+1}}{\mu_{i-1}} - \frac{1}{\mu_{i-1}}\\
&= \frac{1}{\mu_{i-1}} \left( \sum_{j=i}^{n-1} \left( \prod_{k=i}^j p_k \right)  - \sum_{j=i-1}^{n-1} \left( \prod_{k=i}^{j+1} p_k \right)\right)\\
&= -\frac{1}{\mu_{i-1}}\prod_{k=i}^n p_k = 0,
\end{align*}
as $p_n = 0$.

\section{Proof of Proposition 5}\label{app:lemanypi}

Let $\beta$ be the unique invariant vector of $S+(-Se)\alpha$, that is, $\beta (S+(-Se)\alpha)=0$ and
$\beta e = 1$. It is easy to verify that $\beta_i = (\prod_{j=1}^{i-1} p_j)/\mu_i$ (as
the mean service time $\sum_i \beta_i \mu_i$ equals one). 

Due to 
\eqref{eq:newODE2} with $i=2$, we immediately have that
\begin{align}\label{eq:pi11} 
\sum_{j=1}^n (\pi_{1,j}-\pi_{1,j+1})\nu_j = (\pi_{1,1}-\pi_{1,2}) \mu_1.
\end{align}
By means of \eqref{eq:newODE2} we have for $i \geq 2$:
\[\frac{d}{dt} (h_{1,i}(t)-h_{1,i+1}(t)) = (h_{1,i-1}(t)-h_{1,i}(t)) p_{i-1}\mu_{i-1} - (h_{1,i}(t)-h_{1,i+1}(t)) \mu_i, \]
meaning 
\begin{align}\label{eq:pi1i}
(\pi_{1,i-1}-\pi_{1,i}) p_{i-1}\mu_{i-1} = (\pi_{1,i}-\pi_{1,i+1}) \mu_i.
\end{align} 
Combining \eqref{eq:pi11}  and \eqref{eq:pi1i}, we see that the vector
\[(\pi_{1,1}-\pi_{1,2}, \pi_{1,2}-\pi_{1,3}, \ldots, \pi_{1,n-1}-\pi_{1,n}, \pi_{1,n}),\] 
is an invariant vector of $S+(-Se)\alpha$ and thus a multiple of the vector $\beta$.
The result then follows as $\beta e = 1$.

\end{document}